\documentclass[a4paper,UKenglish,cleveref, autoref, thm-restate]{lipics-v2021}

\nolinenumbers
\title{Weighted Ancestors in Suffix Trees Revisited} 

\author{Djamal Belazzougui}{Centre de Recherche sur l'Information Scientifique et Technique, Algiers, Algeria}{dbelazzougui@cerist.dz}{}{}

\author{Dmitry Kosolobov}{Institute of Natural Sciences and Mathematics, Ural Federal University, Ekaterinburg, Russia}{dkosolobov@mail.ru}{https://orcid.org/0000-0002-2909-2952}{Supported by the grant 20-71-00050 of the Russian Foundation of Basic Research (for the final version of the data structure and for its construction algorithm).}

\author{Simon J. Puglisi}{Department of Computer Science, University of Helsinki, Helsinki, Finland}{puglisi@cs.helsinki.fi}{https://orcid.org/0000-0001-7668-7636}{}

\author{Rajeev Raman}{Department of Informatics, University of Leicester, Leicester, United Kingdom}{r.raman@leicester.ac.uk}{https://orcid.org/0000-0001-9942-8290}{}

\authorrunning{D. Belazzougui, D. Kosolobov, S.\,J. Puglisi, R. Raman} 

\Copyright{Djamal Belazzougui, Dmitry Kosolobov, Simon J. Puglisi, and Rajeev Raman} 



\category{} 

\relatedversion{} 

\EventEditors{Pawe{\l} Gawrychowski and Tatiana Starikovskaya}
\EventNoEds{2}
\EventLongTitle{32nd Annual Symposium on Combinatorial Pattern Matching (CPM 2021)}
\EventShortTitle{CPM 2021}
\EventAcronym{CPM}
\EventYear{2021}
\EventDate{July 5--7, 2021}
\EventLocation{Wroc{\l}aw, Poland}
\EventLogo{}
\SeriesVolume{191}
\ArticleNo{35}

\newcommand\Oh{\mathcal{O}}
\newcommand\lcp{\mathop{\mathsf{lcp}}}
\newcommand\plcp{\mathop{\mathsf{plcp}}}
\newcommand\str{\mathop{\mathsf{str}}}
\newcommand\LCP{\mathsf{lcp}}
\newcommand\PLCP{\mathsf{plcp}}
\newcommand\BWT{\mathsf{bwt}}
\newcommand\SA{\mathsf{sa}}
\newcommand\ISA{\mathsf{isa}}

\begin{document}

\maketitle

\begin{abstract}
The weighted ancestor problem is a well-known generalization of the predecessor problem to trees. It is known to require $\Omega(\log\log n)$ time for queries provided $\Oh(n\mathop{\mathrm{polylog}} n)$ space is available and weights are from $[0..n]$, where $n$ is the number of tree nodes. However, when applied to suffix trees, the problem, surprisingly, admits an $\Oh(n)$-space solution with constant query time, as was shown by Gawrychowski, Lewenstein, and Nicholson (Proc. ESA 2014). This variant of the problem can be reformulated as follows: given the suffix tree of a string $s$, we need a data structure that can locate in the tree any substring $s[p..q]$ of $s$ in $\Oh(1)$ time (as if one descended from the root reading $s[p..q]$ along the way). Unfortunately, the data structure of Gawrychowski et~al.~has no efficient construction algorithm, limiting its wider usage as an algorithmic tool. In this paper we resolve this issue, describing a data structure for weighted ancestors in suffix trees with constant query time and a linear construction algorithm. Our solution is based on a novel approach using so-called irreducible LCP values.
\end{abstract}

\section{Introduction}
\label{sec:intro}

The suffix tree is one of the central data structures in stringology. Its primary application is to determine whether an arbitrary string $t$ occurs as a substring of another string, $s$, which can be done in time proportional to the length of $t$ by traversing the suffix tree of $s$ downward from the root and reading off the symbols of $t$ along the way.
Many algorithms using suffix trees perform this procedure for substrings $t = s[p..q]$ of $s$ itself. In this important special case, the traversal can be executed much faster than $\Oh(q - p + 1)$ time provided the tree has been preprocessed to build some additional data structures~\cite{ALLS,FarachMuthukrishnan,GLN}; particularly surprising is that the traversal can be performed in constant time using only linear space, as shown by Gawrychowski et al.~\cite{GLN}.
In this paper, we describe the first linear construction algorithm for such a data structure. The lack of an efficient construction algorithm for the result of~\cite{GLN} has been, apparently, the main obstacle hindering its wider adoption. Our solution is completely different from that of~\cite{GLN} and is based on a combinatorial result of K{\"a}rkk{\"a}inen et al.~\cite{KarkkainenManziniPuglisi} (see also~\cite{KarkkainenKempaPiatkowski}) that estimates the sum of irreducible LCP values (precise definitions follow).

As one might expect, the described data structure has a multitude of applications:~\cite{ALLS,BiswasEtAl,BiswasEtAl2,FarachMuthukrishnan,KempaPolicritiPrezzaRotenberg,KociumakaEtAl,KopelowitzEtAl}, to name a few. We, however, do not dive further into a discussion of these applications and refer the reader to the overview in~\cite[Sect. 1]{GLN} and references therein for more details.


In order to ``locate'' a substring $t = s[p..q]$ in the suffix tree of $s$, it suffices to answer the following query: given a node of the tree that corresponds to the suffix of $s$ starting at position $p$ (usually, it is a leaf), we should find the farthest (closest to the root) ancestor $v$ of this node such that the string written on the root--$v$ path has length at least $q - p + 1$. This problem is a particular case of the general \emph{weighted ancestor problem}~\cite{ALLS,FarachMuthukrishnan,KopelowitzLewenstein}: given a tree whose nodes are associated with integer weights such that the weights decrease if one ascends from any node to the root (the weight of a node in the suffix tree is the length of the string written on the root--node path), the tree should be preprocessed in order to answer \emph{weighted ancestor queries} that, for a given node $v$ and a number $w$, return the farthest ancestor of $v$ whose weight is at least $w$. The problem can be viewed as a generalization to trees of the predecessor search problem, in which we preprocess a set of integers to support predecessor queries: for any given number $w$, return the largest integer from the set that precedes $w$.

Clearly, any linear-space solution for weighted ancestors can be used as a solution for predecessor search. As was shown in~\cite{KopelowitzLewenstein} and~\cite{GLN}, a certain converse reduction is also possible: the weighted ancestor queries for a tree with $n$ nodes and integer weights from a range $[0..U]$ can be answered in $\Oh(\mathrm{pred}(n, U))$ time using linear space, where $\mathrm{pred}(n, U)$ is the time required to answer predecessor search queries for any set of $k \le n$ integers from the range $[0..U]$ using $\Oh(k)$ space. Therefore, when $U = n$ (as in the case of suffix tree), both problems can be solved in linear space with $\Oh(\log\log n)$-time queries using the standard van Emde Boas or y-fast trie data structures~\cite{vanEmdeBoas,Willard}.

Due to the lower bound of P{\u{a}}tra{\c{s}}cu and Thorup~\cite{PatrascuThorup}, the time $\Oh(\log\log n)$ is optimal for $U = n$ when the available space is linear, and moreover, any solution of the weighted ancestor problem that uses $\Oh(n\mathop{\mathrm{polylog}} n)$ space must spend $\Omega(\log\log n)$ time on queries (see~\cite[Appendix A]{GLN}). In view of this lower bound, it is all the more unexpected that the special case of suffix trees admits an $\Oh(n)$-space solution with constant query time. In order to circumvent the lower bound, Gawrychowski et al.~\cite{GLN} solve predecessor search problems on some paths of the suffix tree using $\Oh(n)$ \emph{bits} of space (or slightly less), which admits a constant time solution by a so-called rank data structure; because of the internal repetitive structure of the suffix tree, the solution for one path can be reused in many different paths in such a way that, in total, the utilized space is linear. Our approach essentially relies on the same intuition but we perform path predecessor queries on different trees closely related to the suffix tree and the advantages of repetitive structures come implicitly; in particular, we do not explicitly treat periodic regions of the string separately so that, in this regard, our solution is more ``uniform'', in a sense, than that of~\cite{GLN}.

This paper is organized as follows. In Section~\ref{sec:prelim} the problem is reduced to certain path-counting queries using (unweighted) level ancestor queries. Section~\ref{sec:queries-at-irreducible} describes a simple solution for the queries locating substrings $s[p..q]$ whose position $p$ corresponds to an irreducible LCP value. In Section~\ref{sec:reduction-to-irreducible} we reduce general queries to the queries at irreducible positions and a~certain geometric orthogonal predecessor problem. Section~\ref{sec:special-weighted-ancestors} describes a solution for this special geometric problem. Conclusions and reflections are then offered in Section~\ref{sec:conclusion}.

\section{Basic Data Structures}
\label{sec:prelim}

Let us fix a string $s$ of length $n$, denoting its letters by $s[0], s[1], \ldots, s[n{-}1]$. We write $s[p..q]$ for the \emph{substring} $s[p]s[p{+}1]\cdots s[q]$, assuming it is empty if $p > q$; $s[p..q]$ is called a \emph{suffix} (resp., \emph{prefix}) of $s$ if $q = n - 1$ (resp., $p = 0$). For any string $t$, let $|t|$ denote its length. We say that $t$ \emph{occurs} at position $p$ in $s$ if $s[p..p{+}|t|{-}1] = t$. Denote $[p..q] = \{k \in \mathbb{Z} \colon p\le k\le q\}$.

A \emph{suffix tree} of $s$ is a compacted trie containing all suffixes of $s$. The labels on the edges of the tree are stored as pointers to corresponding substrings of $s$. For each tree node $v$, denote by $\str(v)$ the string written on the root--$v$ path. The number $|\str(v)|$ is called the \emph{string depth} of $v$. The \emph{locus} of a substring $s[p..q]$ is the (unique) node $v$ of minimal string depth such that $s[p..q]$ is a prefix of $\str(v)$.

The string $s$ and its suffix tree $T$ are the input to our algorithm. To simplify the exposition, we assume that $s[n{-}1]$ is equal to a special letter $\$$ that is smaller than all other letters in $s$, so that there is a one-to-one correspondence between the suffixes of $s$ and the leaves of~$T$. Our computational model is the word-RAM and space is measured in $\Theta(\log n)$-bit machine words. Our goal is to construct in $\Oh(n)$ time a data structure that can find in the tree the locus of any substring $s[p..q]$ of $s$ in $\Oh(1)$ time.

A \emph{level ancestor query} in a tree asks, for a given node $v$ and an integer $d \ge 1$, the $d$th node on the $v$--root path (provided the path has at least $d$ nodes). It is known that any tree can be preprocessed in linear time to answer such queries in constant time~\cite{BenderFarachColton,BerkmanVishkin}. With such a structure, the locus of a substring $s[p..q]$ can be found by first locating the leaf $v$ of $T$ corresponding to $s[p..n{-}1]$ (i.e., $\str(v) = s[p..n{-}1]$, which can be located using a precomputed array of length $n$) and, then, counting the number $d$ of nodes $u$ on the $v$--root path such that $|\str(u)| > q - p$; then, evidently, the locus of $s[p..q]$ is given by the level ancestor query on $v$ and $d$.

We have to complicate this scheme slightly since the machinery that we develop in the sequel allows us to count only those nodes $u$ on the $v$--root path that have a branch to the ``left'' of the path (or, symmetrically, to the ``right''). More formally, given a node $v$, a node $u$ on the $v$--root path is called \emph{left-branching} (resp., \emph{right-branching}) if there is a suffix $s[p..n{-}1]$ that is lexicographically smaller (resp., greater) than the string $\str(v)$ and its longest common prefix with $\str(v)$ is $\str(u)$. For instance, the path from the leaf corresponding to the string $sippi\$$ in Figure~\ref{fig:suffix-tree} has four nodes but only one of them (namely, the root) is left-branching.

\begin{lemma}
For any suffix tree, one can build in linear time a data structure that, for any node $v$ and integer $d \ge 1$, can return in $\Oh(1)$ time the $d$th left-branching (or right-branching) node on the $v$--root path.\label{lem:left-branch-la}
\end{lemma}
\begin{proof}
Traversing the suffix tree, we construct another tree on the same set of nodes in which the parent of each non-root node is either its nearest left-branching ancestor in the suffix tree (if any) or the root. The queries for left-branching nodes can be answered by the level ancestor structure~\cite{BenderFarachColton,BerkmanVishkin} built on this new tree.  The right-branching case is symmetric.
\end{proof}

Thus, to find the locus for $s[p..q]$, we have to count  on a leaf--root path the number of left- and right-branching nodes whose string depths are greater than the threshold $q - p$; we then use these two numbers in the data structure of Lemma~\ref{lem:left-branch-la} in order to find two candidate nodes and the node with smaller string depth is the locus. In the remaining text, we focus only on this counting problem and only on left-branching nodes as the right-branching case is symmetric. First, however, we define a number of useful standard structures.

The \emph{suffix array} of $s$ is an array $\SA[0..n{-}1]$ containing integers from $0$ to $n-1$ such that $s[\SA[0] .. n{-}1] < s[\SA[1] .. n{-}1] < \cdots < s[\SA[n{-}1] .. n{-}1]$ lexicographically~\cite{ManberMyers}. The \emph{inverse suffix array}, denoted $\ISA[0..n{-}1]$, is defined as $\SA[\ISA[p]] = p$, for all $p \in [0..n{-}1]$. For any positions $p$ and $q$, denote by $\lcp(p, q)$ the length of the longest common prefix of $s[p..n{-}1]$ and $s[q..n{-}1]$. The \emph{longest common prefix (LCP) array} is an array $\LCP[0..n{-}1]$ such that $\LCP[0] = 0$ and $\LCP[i] = \lcp(\SA[i{-}1], \SA[i])$, for $i \in [1..n{-}1]$. The \emph{permuted longest common prefix (PLCP) array} is an array $\plcp[0..n{-}1]$ such that $\plcp[p] = \LCP[\ISA[p]]$, for $p \in [0..n{-}1]$. The \emph{Burrows--Wheeler transform}~\cite{BurrowsWheeler} is a string $\BWT[0..n{-}1]$ such that $\BWT[i] = s[\SA[i]{-}1]$ if $\SA[i] \ne 0$, and $\BWT[i] = s[n{-}1]$ otherwise. All the arrays $\SA$, $\ISA$, $\LCP$, $\plcp$ and the string $\BWT$ can be built from $T$ in $\Oh(n)$ time. Some of these structures are depicted in Figure~\ref{fig:suffix-tree}.

\begin{figure}[htb]
\centering
\begin{subfigure}[t]{.55\textwidth}
\includegraphics[scale=1]{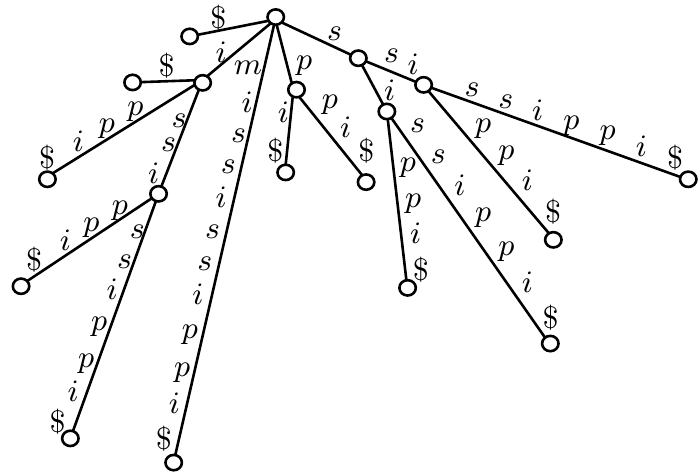}
\end{subfigure}
\hfill
\begin{subtable}[t]{.42\textwidth}
\small
\vskip-50mm
\begin{tabular}{crrrl}
$i$ & $\LCP$ & $\SA$  & $\BWT$ & sorted suffixes \\\cline{1-5}
0 & \textbf 0 & 11 & $i$ & $\$$    \\
1 & \textbf 0 & 10 & $p$ & $i\$$    \\
2 & \textbf 1 & 7  & $s$ & $ippi\$$   \\
3 &     1 & 4  & $s$ & $issippi\$$   \\
4 & \textbf 4 & 1  & $m$ & $ississippi\$$  \\
5 & \textbf 0 & 0  & $\$$ & $mississippi\$$ \\
6 & \textbf 0 & 9  & $p$ & $pi\$$    \\
7 & \textbf 1 & 8  & $i$ & $ppi\$$   \\
8 & \textbf 0 & 6  & $s$ & $sippi\$$  \\
9 &     2 & 3  & $s$ & $sissippi\$$   \\
10 & \textbf 1 & 5  & $i$ & $ssippi\$$ \\
11 &     3 & 2  & $i$ & $ssissippi\$$  \\
\end{tabular}
\end{subtable}
\caption{The suffix tree $T$, $\LCP$, $\SA$, and $\BWT$ of the string $s = mississippi\$$. All irreducible LCP values are in bold and their sum is $7$; all positions of $s$ except $2, 3, 4$ are irreducible.}
\label{fig:suffix-tree}
\end{figure}

\section{Queries at Irreducible Positions}
\label{sec:queries-at-irreducible}

An LCP value $\LCP[i]$ is called \emph{irreducible} if either $i = 0$ or $\BWT[i{-}1] \ne \BWT[i]$. In other words, the irreducible LCP values are those values of $\LCP[0..n{-}1]$ that correspond to the first positions in the runs of $\BWT$ (e.g., all values $\LCP[i]$ except for $i = 3,9,11$ in Figure~\ref{fig:suffix-tree}). The central combinatorial fact that is at the core of our construction is the following surprising lemma proved by K{\"a}rkk{\"a}inen et~al.~\cite{KarkkainenManziniPuglisi}.

\begin{lemma}[see \cite{KarkkainenManziniPuglisi}]
The sum of all irreducible LCP values is at most $2 n\log n$.
\label{lem:irreducible-lcp}
\end{lemma}

We say that a position $p$ of the string $s$ is \emph{irreducible} if $\LCP[\ISA[p]]$ (${=}\PLCP[p]$) is an irreducible LCP value.

Consider a query that asks to find the locus of a substring $s[p..q]$ such that $p$ is an irreducible position. Denote $\ell = q - p + 1$. We first locate in $\Oh(1)$ time the leaf $v$ corresponding to $s[p..n{-}1]$ using a precomputed array. As was discussed above, the problem is reduced to the counting of left-branching nodes $u$ on the $v$--root path such that $|\str(u)| \ge \ell$. We will construct for $p$ a bit array $b_p[0..\ell_p{-}1]$ of length $\ell_p = \PLCP[p]$ such that, for any $d$, $b_p[d] = 1$ iff there is a left-branching node with string depth $d$ on the $v$--root path. Since the suffix $s[p..n{-}1]$ and its lexicographical predecessor among all suffixes (if any, i.e., if $\ISA[p] \ne 0$) have the longest common prefix of length $\PLCP[p]$, the lowest left-branching node on the $v$--root path has string depth $\PLCP[p]$. Therefore, the number of left-branching nodes $u$ on the $v$--root path such that $|\str(u)| \ge \ell$ is equal to the number of ones in the subarray $b_p[\ell..\ell_p{-}1]$ plus~$1$. The bit counting query is answered using the following \emph{rank data structure}.

\begin{lemma}[see \cite{Clark,Jacobson}]
For any bit array $b[0..m{-}1]$ packed into $\Oh(m / w)$ $w$-bit machine words, one can construct in $\Oh(m / w)$ time a rank data structure that can count the number of ones in any range $b[p..q]$ in $\Oh(1)$ time, provided a table computable in $o(2^w)$ time and independent of the array has been precomputed.\label{lem:rank-select}
\end{lemma}

By Lemma~\ref{lem:irreducible-lcp}, the total length of all arrays $b_p$ for all irreducible positions $p$ of $s$ is $\Oh(n\log n)$. Therefore, one can construct the rank data structures of Lemma~\ref{lem:rank-select} for them in $\Oh(n)$ overall time provided $w = \lfloor\log n\rfloor$. It remains to build the arrays $b_p$ themselves.

To this end, we traverse the suffix tree $T$ in depth first (i.e., lexicographical) order, maintaining along the way a bit array $b$ of length $n$ such that, when we are in a node $v$, the subarray $b[0..|\str(v)|{-}1]$ marks the string depths of all left-branching ancestors of $v$, i.e., for $0 \le d < |\str(v)|$, we have $b[d] = 1$ iff there is a left-branching node with string depth $d$ on the $v$--root path. The array is stored in a packed form in $\lfloor\log n\rfloor$-bit machine words. For each node $v$, we consider its outgoing edges. Each pair of adjacent edges (in the lexicographical order of their labels) corresponds to a unique value $\LCP[i]$ such that $s[\SA[i]..n{-}1]$ is the suffix corresponding to the leftmost (lexicographically smallest) leaf of the subtree connected to the larger one of the two edge labels. We check whether the value $\LCP[i]$ is irreducible comparing $\BWT[i{-}1]$ and $\BWT[i]$ and, if so, store the subarray $b[0..|\str(v)|{-}1]$ into $b_{\SA[i]}$ in $\Oh(1 + |\str(v)| / \log n)$ time. As we touch in this way every value $\LCP[i]$ only once, the overall time is $\Oh(n)$ by the same argument using Lemma~\ref{lem:irreducible-lcp}.

\section{Reduction to Irreducible Positions}
\label{sec:reduction-to-irreducible}

Consider a query for the locus of a substring $s[p..q]$ such that the position $p$ is not irreducible. Let $v$ be the leaf corresponding to $s[p..n{-}1]$. As was discussed in Section~\ref{sec:prelim}, to answer the query, we have to count the number of left-branching nodes with string depths at least $q - p + 1$ on the $v$--root path. As a first approach, we do the following precalculations for this.

We build in $\Oh(n)$ time on the suffix tree $T$ a data structure that allows us to find the lowest common ancestor for any pair of nodes in $\Oh(1)$ time~\cite{BenderFarachColton2,HarelTarjan}. In one tree traversal, we precompute in each node $u$ the number of left-branching nodes on the $u$--root path. We create two arrays $R_{\le}[0..n{-}1]$ and $R_{\ge}[0..n{-}1]$ such that, for any $i$, $R_{\le}[i] = \max\{j \le i \colon j = 0\text{ or }\BWT[j{-1}] \ne \BWT[j]\}$ and $R_{\ge}[i] = \min\{j \ge i \colon j = 0\text{ or } \BWT[j{-1}] \ne \BWT[j]\}$ ($R_{\ge}[i]$ is undefined if there is no such $j$), i.e., $R_{\le}[i]$ and $R_{\ge}[i]$ are indexes of irreducible LCP values, respectively, preceding and succeeding $\LCP[i]$. Thus, for any suffix $s[p..n{-}1]$, the suffixes $s[\SA[R_{\le}[\ISA[p]]]..n{-}1]$ and $s[\SA[R_{\ge}[\ISA[p]]]..n{-}1]$ are, respectively, the closest lexicographical predecessor and successor of $s[p..n{-}1]$ with irreducible starting position. The \emph{closest irreducible lexicographical neighbour of $s[p..n{-}1]$} is that suffix $s[r..n{-}1]$ among these two that has longer common prefix with $s[p..n{-}1]$, i.e., $r = \SA[R_{\le}[\ISA[p]]]$ if either $\lcp(\SA[R_{\le}[\ISA[p]]], p) \ge \lcp(\SA[R_{\ge}[\ISA[p]]], p)$ or $R_{\ge}[\ISA[p]]$ is undefined, and $r = \SA[R_{\ge}[\ISA[p]]]$ otherwise. Since $\lcp(t, p)$ can be calculated in $\Oh(1)$ time for any positions $t$ and $p$ by one lowest common ancestor query on their corresponding leaves, the closest irreducible lexicographical neighbour can be computed in $\Oh(1)$ time.

Now, consider a query for the locus of $s[p..q]$ with non-irreducible $p$. We first find in $\Oh(1)$ time the closest irreducible lexicographical neighbour $s[r..n{-}1]$ for $s[p .. n{-}1]$. Let $v'$ and $v$ be the leaves corresponding to $s[r..n{-}1]$ and $s[p..n{-}1]$, respectively. We compute in $\Oh(1)$ time the lowest common ancestor $u$ of $v'$ and $v$. Thus, $|\str(u)| = \lcp(r, p)$. Using the number of left-branching nodes on the $v$--root and $u$--root paths that were precomputed in $v$ and $u$, we calculate in $\Oh(1)$ time the number $k$ of left-branching nodes on the $v$--root path that lie between the nodes $v$ and $u$ (inclusively).

Denote $\ell = q - p + 1$. If $\ell \le |\str(u)|$, then we can count in the $v$--root path the number of left-branching nodes $w$ such that $|\str(w)| \ge \ell$ as follows. The number of nodes $w$ such that $|\str(w)| \ge |\str(u)|$ is equal to $k$. The number of nodes $w$ such that $\ell \le |\str(w)| < |\str(u)|$ can be found by counting the number of ones in the subarray $b_r[\ell..|\str(u)|{-}1]$ of the bit array $b_r$ associated with the irreducible position $r$ (assuming that all values $b_r[t]$ with $t \ge \ell_r$ are zeros in case the length $\ell_r$ of $b_r$ is less than $|\str(u)|$), which can be performed in $\Oh(1)$ time by Lemma~\ref{lem:rank-select}. Thus, the problem is solved for the case $\ell \le |\str(u)|$.

To address the case $\ell > |\str(u)|$, we have to develop more sophisticated techniques that allow us to reduce the counting of left-branching nodes on the $v$--root path to counting on a different leaf--root path with a different threshold $\ell'$ that meets the condition $\ell' \le |\str(u')|$, for an analogously appropriately defined node $u'$. The remainder of the text describes the reduction.

\medskip

Similar to irreducible positions, let us define, for each non-irreducible position $p$, a number $\ell_p = \PLCP[p]$ and an array $b_p[0..\ell_p{-}1]$ such that, for any $d$, $b_p[d] = 1$ iff there is a left-branching node of string depth $d$ on the path from the leaf corresponding to $s[p..n{-}1]$ to the root. We do not actually store the arrays $b_p$ and use them only in the analysis.

\begin{lemma}
For any non-irreducible position $p$, we have $\ell_{p-1} = \ell_p + 1$ and, if $b_{p-1}[d{+}1] = 1$ for some $d \ge 0$, then $b_p[d] = 1$.\label{lem:inherent-arrays}
\end{lemma}
\begin{proof}
Let $s[t..n{-}1]$ be the suffix lexicographically preceding $s[p..n{-}1]$, i.e., $t = \SA[\ISA[p]{-}1]$ and $\ell_p = \lcp(t, p)$. Since $p$ is not irreducible, we have $\BWT[\ISA[t]] = \BWT[\ISA[p]]$, i.e., $s[t{-}1] = s[p{-}1]$. Hence, $s[t{-}1..n{-}1]$ is the suffix lexicographically preceding $s[p{-}1..n{-}1]$: $\ISA[t{-}1] = \ISA[p{-}1] - 1$. Therefore, $\ell_{p-1} = \lcp(t - 1, p - 1)$, which is equal to $\lcp(t, p) + 1 = \ell_p + 1$.

If $b_{p-1}[d{+}1] = 1$ for $d \ge 0$, then there is a suffix $s[r..n{-}1]$ that is smaller than $s[p{-}1..n{-}1]$ and $\lcp(r, p - 1) = d + 1$. Then, $\lcp(r + 1, p) = d$, which implies that $b_p[d] = 1$.
\end{proof}

Consider two consecutive irreducible positions $p'$ and $p''$ in $s$ (i.e., all positions $r$ between $p'$ and $p''$ are not irreducible). Lemma~\ref{lem:inherent-arrays} states that the arrays $b_{p'}, b_{p'+1}, \ldots, b_{p''-1}$ form a trapezoidal structure as depicted in Figure~\ref{fig:trapezoid} in which each $1$ value is inherited by all arrays above it while it fits in their range.

\begin{figure}[htb]
\centering
\includegraphics[scale=1]{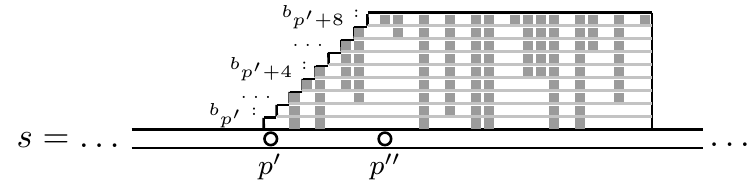}
\caption{Here, $p'$ and $p''$ are consecutive irreducible positions. Each line of the trapezoid depicts an array $b_{p}[0..\ell_{p}{-}1]$, for $p \in [p'..p''{-}1]$, in which gray and white positions signify, respectively, ones and zeros.}\label{fig:trapezoid}
\end{figure}

The query for the locus of $s[p..q]$ was essentially reduced to the counting of ones in the subarray $b_p[\ell..\ell_p{-}1]$, where $\ell = q - p + 1$. The problem now is that the array $b_p$ is not stored explicitly since the position $p$ is not irreducible. Denote the length of $b_p[\ell..\ell_p{-}1]$ by $m = \ell_p - \ell$, so that $b_p[\ell..\ell_p{-}1] = b_p[\ell_p{-}m..\ell_p{-}1]$, which is a more convenient notation for what follows. Let $p'$ be the closest irreducible position preceding $p$, i.e., $p' = \max\{r \le p \colon r\text{ is irreducible}\}$. If we are lucky and neither of the arrays $b_{p'+1}, b_{p'+2}, \ldots, b_p$ introduces new $1$ values in the subarray $b_{p'}[\ell_{p'}{-}m..\ell_{p'}{-}1]$ (in other words if $b_{p'}[\ell_{p'}{-}m..\ell_{p'}{-}1] = b_p[\ell_p{-}m..\ell_p{-}1]$), then we can simply count the number of ones in $b_{p'}[\ell_{p'}{-}m..\ell_{p'}{-}1]$ in $\Oh(1)$ time using Lemma~\ref{lem:rank-select} and, thus, solve the problem. Unfortunately, new $1$ values indeed could be introduced. But, as it turns out, such new values are, in a sense, ``covered'' by other arrays $b_r$ at some irreducible positions $r$ as the following lemma suggests.

\begin{lemma}
Let $p$ be a non-irreducible position and $s[r..n{-}1]$ be the closest irreducible lexicographical neighbour of $s[p..n{-}1]$. Denote $c_p = \lcp(r,p)$. If, for some $d$, we have $b_p[d] = 1$ but $b_{p-1}[d{+}1] = 0$, then $d \le c_p$.\label{lem:new-branch}
\end{lemma}
\begin{proof}
The condition $b_p[d] = 1$ implies that there is a suffix $s[t..n{-}1]$ that is smaller than $s[p..n{-}1]$ and $\lcp(t, p) = d$. If $\BWT[\ISA[t]] = \BWT[\ISA[p]]$, then $s[t{-}1] = s[p{-}1]$ and, hence, the suffix $s[t{-}1..n{-}1]$ is smaller than $s[p{-}1..n{-}1]$ and $\lcp(t - 1, p - 1) = d + 1$. This gives $b_{p-1}[d{+}1] = 1$, which contradicts the equality $b_{p-1}[d{+}1] = 0$. Hence, $\BWT[\ISA[t]] \ne \BWT[\ISA[p]]$. Thus, at least one of the values $\LCP[\ISA[t]{+}1], \LCP[\ISA[t]{+}2], \ldots, \LCP[\ISA[p]{-}1]$ must be irreducible (note that $\ISA[t] + 1 < \ISA[p]$ since otherwise $p$ would be irreducible). Let $r'$ be an irreducible position such that $\ISA[t] < \ISA[r'] < \ISA[p]$. We obtain $\lcp(r, p) \ge d$ since $\lcp(r', p) \ge \lcp(t, p) = d$ and, for any irreducible $r''$ (in particular, for $r'$), we have $\lcp(r, p) \ge \lcp(r'', p)$.
\end{proof}

Observe that $c_p$ in Lemma~\ref{lem:new-branch} is equal to $|\str(u)|$, where $u$ is the lowest common ancestor of the leaves $v'$ and $v$ corresponding to $s[r..n{-}1]$ and $s[p..n{-}1]$. Thus, the numbers $c_p$ can be precomputed for all non-irreducible positions $p$ in $\Oh(n)$ time using lowest common ancestor queries and the arrays $R_{\le}$ and $R_{\ge}$ described in the beginning of the present section. For irreducible positions $p'$, we put $c_{p'} = \ell_{p'}$ by definition. To illustrate Lemma~\ref{lem:new-branch} and its consequences, we depict in Figure~\ref{fig:trapezoid-shaded} the same trapezoidal structure as in Figure~\ref{fig:trapezoid} coloring in each $b_p$ a prefix of length $c_p + 1$ (and also depicting the range $b_p[\ell_p{-}m .. \ell_p{-}1]$).

\begin{figure}[htb]
\centering
\includegraphics[scale=1]{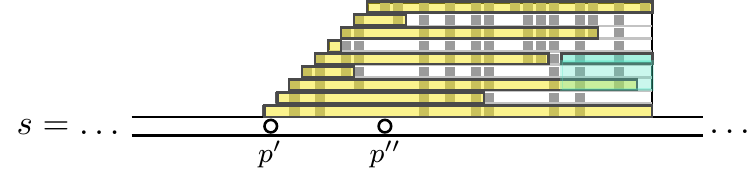}
\caption{The yellow stripe in each array $b_p$ is a prefix of length $\min\{c_p + 1, \ell_p\}$. Note that, according to Lemma~\ref{lem:new-branch}, each gray position whose corresponding position below is white is covered in yellow. The blue stripe is $b_p[\ell_p{-}m .. \ell_p{-}1]$, where $p = p' + 4$. It coincides with two corresponding regions below it that are in light blue. Here we have $t = p - 2$.}\label{fig:trapezoid-shaded}
\end{figure}

The problematic condition $\ell > |\str(u)|$ that we consider can be reformulated as $\ell_p - m > c_p$. It follows from Lemmas~\ref{lem:inherent-arrays} and~\ref{lem:new-branch} that in this case $b_p[\ell_p{-}m..\ell_p{-}1] = b_{p-1}[\ell_{p-1}{-}m..\ell_{p-1}{-}1]$. Let $t$ be the first position preceding $p$ such that $\ell_{t} - m \le c_{t}$. Note that $t \ge p'$ since $c_{p'} = \ell_{p'}$ by definition and, thus, $\ell_{p'} - m \le c_{p'}$. Then, applying Lemma~\ref{lem:new-branch} consecutively to all positions $p, p-1, \ldots, t + 1$, we conclude that $b_{t}[\ell_{t}{-}m .. \ell_{t}{-}1] = b_p[\ell_p{-}m .. \ell_p{-}1]$.

Informally, in terms of Figure~\ref{fig:trapezoid-shaded}, the searching of $t$ corresponds to the moving of the range $b_p[\ell_p{-}m .. \ell_p{-}1]$ down until we encounter an ``obstacle'', a colored part of $b_{t}$ of length $c_{t} + 1$. The specific algorithm finding $t$ is discussed in Section~\ref{sec:special-weighted-ancestors}; let us assume for the time being that $t$ is already known. Then, Lemma~\ref{lem:new-branch} implies that all the ranges $b_r[\ell_r{-}m .. \ell_r{-}1]$ with $r \in [t..p]$ coincide and, thus, the whole problem was reduced to the counting of ones in the subarray $b_t[\ell_{t}{-}m .. \ell_{t}{-}1]$. But since $\ell_{t} - m \le c_{t}$, the problem can be solved by the method described at the beginning of the section: we find an irreducible position $r$ such that $\lcp(r, t) = c_{t}$, then locate the leaves $v'$ and $v''$ corresponding to $s[r .. n{-}1]$ and $s[t .. n{-}1]$, respectively, find the lowest common ancestor $u'$ of $v'$ and $v''$, and separately count the number of left-branching nodes $w$ on the $v''$--root path such that $|\str(w)| \ge c_{t}$ and such that $\ell_{t} - m \le |\str(w)| < c_{t}$.

Let us briefly recap the reductions described above: the searching for the locus of $s[p..q]$ was essentially reduced to the counting of ones in the subarray $b_p[\ell_p{-}m .. \ell_p{-}1]$, where $m = \ell_p - (q - p + 1)$, for which we first compute the position $t = \max\{t \le p \colon \ell_t - m \le c_t\}$ (this is discussed in Section~\ref{sec:special-weighted-ancestors}) and, then, count the number of ones in the subarray $b_t[\ell_t{-}m .. \ell_t{-}1]$ (the subarray coincides with $b_p[\ell_p{-}m .. \ell_p{-}1]$ by Lemmas~\ref{lem:inherent-arrays} and~\ref{lem:new-branch}) using lowest common ancestor queries, the arrays $R_{\le}$ and $R_{\ge}$, and some precomputed numbers in the nodes.

\section{Reduction to Special Weighted Ancestors}
\label{sec:special-weighted-ancestors}

We are to compute $t = \max\{t \le p \colon \ell_t - m \le c_t\}$, for a non-irreducible position $p$. As is seen in Figure~\ref{fig:trapezoid-shaded}, this is a kind of geometric ``orthogonal range predecessor'' problem.

For each irreducible position $p'$ in $s$, we build a tree $I_{p'}$ with $p'' - p'$ nodes, where $p''$ is the next irreducible position after $p'$ (so that all $r$ with $p' < r < p''$ are not irreducible): Each node of $I_{p'}$ is associated with a unique position from $[p'..p''{-}1]$ and the root is associated with $p'$. A node associated with position $r \ne p'$ has weight $w_r = r - p' + c_r$; the root has weight $w_{p'} = +\infty$. The parent of a non-root node associated with $r$ is the node associated with the position $r' = \max\{r' < r \colon w_{r'} > w_r\}$. Thus, the weights strictly increase as one ascends to the root. The tree $I_{p'}$ is easier to explain in terms of the trapezoidal structure drawn in Figure~\ref{fig:trapezoid-tree}: the parent of a node associated with $r$ is its closest predecessor $r'$ whose colored range $b_{r'}[r'..r'{+}c_{r'}{-}1]$ goes at least one position farther to the right than $b_{r}[r..r{+}c_{r}{-}1]$. By Lemma~\ref{lem:inherent-arrays}, $r - p' + \ell_r = \ell_{p'}$ and, hence, each weight $w_r$ such that $c_r \le \ell_r$ is upperbounded by $\ell_{p'}$. However, it may happen that $c_r > \ell_r$ and, thus, $w_r > \ell_{p'}$. In this case, the colored range $b_r[r..r{+}c_{r}{-}1]$ stretches beyond the length $\ell_r$ of $b_r$ (for simplicity, Figure~\ref{fig:trapezoid-tree} does not have such cases). Such large weights are a source of technical complications in our scheme.



\begin{figure}[htb]
\centering
\includegraphics[scale=1]{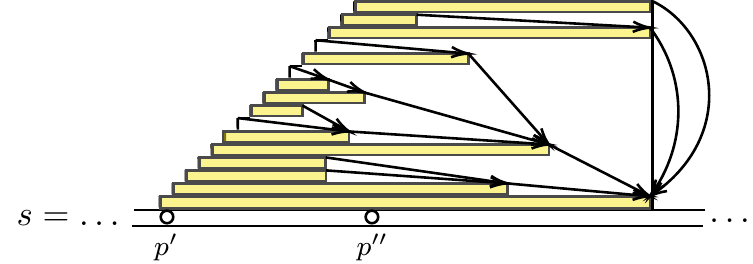}
\caption{Each stripe of the trapezoid depicts an array $b_r$, for $r \in [p'..p''{-}1]$, and its prefix of length $c_r$ (possibly empty) is colored in yellow (not $\min\{c_r + 1, \ell_r\}$ like in Figure~\ref{fig:trapezoid-shaded}). Each yellow prefix corresponds to a node of $I_{p'}$ and it is connected to the yellow prefix corresponding to its parent in $I_{p'}$. To avoid overloading the picture, gray positions signifying $1$s in $b_r$ are not drawn.}\label{fig:trapezoid-tree}
\end{figure}

The tree can be constructed in linear time inserting consecutively the nodes associated with $p', p'+1, \ldots$ and maintaining a stack that contains all nodes of the path from the last inserted node to the root: to insert a new node with weight $w$, we pop from the stack all nodes until a node heavier than $w$ is encountered to which the new node is attached.

In terms of the tree $I_{p'}$, the node associated with the sought position $t$ is the nearest ancestor of the node associated with $p$ such that $w_t \ge \ell + p - p'$, where $\ell = q - p + 1$: the condition $w_t \ge \ell + p - p'$ is equivalent to $\ell_t - m \le c_t$ since $w_t = t - p' + c_t$, $m = \ell_p - \ell$, and $\ell_t - \ell_p = p - t$. Since the threshold $\ell + p - p'$ in the condition does not exceed $\ell_{p'}$, we will be able to ignore differences between weights larger than $\ell_{p'}$ in our algorithm by ``cutting'' them in a way. Still, even with this restriction, at first glance this looks like quite a general weighted ancestor problem that admits no constant time solution in the space available. However, we will be able to use a structure common to such trees in order to speed up the computation. The idea is to decompose the tree into heavy paths (see below), as is usually done for weighted ancestor queries, and perform fast queries on the paths via the use of more memory; the trick is that, with some care, this space can be shared among ``similar'' trees and can be ``traded'' for irreducible LCP values relying again on Lemma~\ref{lem:irreducible-lcp}.

\newcommand\treesize{\mathop{\mathrm{size}}}

Consider the tree $I_{p'}$. An edge $(v,u)$ connecting a child $v$ to its parent $u$ is called \emph{heavy} if $\treesize(v) > \treesize(u) / 2$, and \emph{light} otherwise, where $\treesize(w)$ denotes the number of nodes in the subtree rooted at $w$. Thus, at most one child can be connected to $u$ by a heavy edge. One can easily show that any leaf--root path contains at most $\log n$ light edges. All nodes of $I_{p'}$ are decomposed into disjoint \emph{heavy paths}, maximal paths with only heavy edges in them. Note that in this version of the heavy-light decomposition~\cite{SleatorTarjan} the number of heavy edges incident to a given non-leaf node can be zero (it then forms a singleton path). The decomposition can be constructed in linear time in one tree traversal.

The \emph{predecessor search problem}, for a set of increasing numbers $w_1, w_2, \ldots, w_k$ and a given threshold $w \le w_k$, is to find $\min\{w_i \colon w \le w_i\}$. Although the problem, in fact, searches for a successor, we call it ``predecessor search'' as it is essentially an equivalent problem.

\begin{lemma}[{\cite[Lem. 11]{GLN}}]
Consider a tree on $m$ nodes with weights from $[0..n]$ such that some of the nodes are marked, any leaf--root path contains at most $\Oh(\log n)$ marked nodes, and the weights on any leaf--root path increase. One can preprocess the tree in $\Oh(m)$ time so that predecessor search can be performed among the marked ancestors of any node in $\Oh(1)$ time, provided a table computable in $o(n)$ and independent of the tree is precalculated.\footnote{Although Lemma 11 in \cite{GLN} does not claim linear construction time, it easily follows from its proof.}\label{lem:log-height-tree}
\end{lemma}

For each heavy path, we mark the node closest to the root. Then, the data structure of Lemma~\ref{lem:log-height-tree} is built on each tree $I_{p'}$, which takes $\Oh(n)$ total time for all the trees $I_{p'}$. To answer a weighted ancestor query on $I_{p'}$, we find in $\Oh(1)$ time the heavy path containing the answer using this data structure and, then, consider the predecessor problem inside the path.

Consider a heavy path whose node weights are $w_{i_1}, w_{i_2}, \ldots, w_{i_k}$ in increasing order. We have to answer a predecessor query on the path for a threshold $w$ such that $w \le \ell_{p'}$ (recall that only thresholds ${\le}\ell_{p'}$ are of interest for us) and it is known that its result is in the path, i.e., $w \le w_{i_k}$. Let $w_{i_m} = \max\{w_{i_j} \colon w_{i_j} \le \ell_{p'}\}$. A~trivial constant-time solution for such queries is to construct a bit array $a[0..c]$, where $c = w_{i_{m}} - w_{i_1}$, endowed with a rank data structure such that, for any $d$, we have $a[d] = 1$ iff $d = w_{i_j} - w_{i_1}$, for some $j \in [1..m]$. Then, the predecessor query with a threshold $w$ such that $w \le \min\{\ell_{p'}, w_{i_k}\}$ can be answered by counting the number $h$ of 1s in the subarray $a[0..w - w_{i_1} - 1]$, thus giving the result $w_{i_{h+1}} = \min\{w_{i_j} \colon w \le w_{i_j}\}$. The array $a$  occupies $\Oh(1 + \ell_{p'} / \log n)$ space since $c \le \ell_{p'}$.

It turns out that, with minor changes, the arrays $a$ can be shared among many heavy paths in different trees. However, this approach \emph{per se} leads to $\Oh(n\log n)$ time and space, as will be evident in the sequel. We therefore need a slightly more elaborate solution.

Instead of the array $a[0..c]$, we construct a bit array $\hat{a}[0..\lfloor c / \lfloor\log n\rfloor\rfloor]$ endowed with a rank data structure such that, for any $d$, $\hat{a}[d] = 1$ iff the subarray $a[d\lfloor\log n\rfloor..(d{+}1)\lfloor\log n\rfloor{-}1]$ contains non-zero values (assuming that $a[i] = 0$, for $i > c$). The bit array $\hat{a}$ packed into $\Oh(1 + c / \log^2 n)$ machine words of size $\lfloor\log n\rfloor$ bits can be built from the numbers $w_{i_1}, w_{i_2}, \ldots, w_{i_k}$ in one pass in $\Oh(k + c / \log^2 n)$ time. For each 1 value, $\hat{a}[d] = 1$, we collect all weights $w_{i_j}$ such that $d\lfloor\log n\rfloor \le w_{i_j} - w_{i_1} < (d{+}1)\lfloor\log n\rfloor$ into a set $S_d$. The sets $S_d$, for all $d$ such that $\hat{a}[d] = 1$, are disjoint and non-empty, and can be assembled in one pass through the weights in $\Oh(k)$ time. We also store pointers $P_h$, with $h = 1,2,\ldots$ ($h\le m$), such that $P_h$ refers to a non-empty set $S_d$ such that $\hat{a}[d] = 1$ is the $h$th $1$ value in $\hat{a}$ (i.e., $h$ is the number of 1s in the subarray $\hat{a}[0..d]$). Each set $S_d$ is equipped with the following \emph{fusion heap} data structure.

\begin{lemma}[see \cite{FredmanWillard,PatrascuThorup2}]
For any set $S$ of $\Oh(\log n)$ integers from $[0..n]$, one can build in $\Oh(|S|)$ time a~fusion heap that answers predecessor queries in $\Oh(1)$ time.
\label{lem:fusion-heap}
\end{lemma}

With this machinery, a predecessor query with a threshold $w \le \min\{\ell_{p'}, w_{i_k}\}$ can be answered by first counting the number $h$ of 1s in the subarray $\hat{a}[0..\lfloor (w - w_{i_1}) / \lfloor\log n\rfloor\rfloor{-}1]$ and, then, finding in $\Oh(1)$ time the predecessor of $w$ in the fusion heap referred by $P_{h+1}$.

The array $\hat{a}$ occupies $\Oh(1 + c / \log^2 n)$ space, which is $\Oh(1 + \ell_{p'} / \log^2 n)$ since $c \le \ell_{p'}$. The computation of $\hat{a}$, which takes $\Oh(k + \ell_{p'} / \log^2 n)$ time, is the most time and space consuming part of the described construction; all other structures take $\Oh(k)$ time and space. We are to show that the computation of $\hat{a}$ can sometimes be avoided if (almost) the same array was already constructed for a different path. The following lemma is the key for this optimization.

\begin{lemma}
Given a tree $I_{p'}$ for an irreducible position $p'$, consider its node $x$ associated with a non-irreducible position $p > p'$. Let $s[r..n{-}1]$ be the closest irreducible lexicographical neighbour of $s[p..n{-}1]$ and let $c_p = \lcp(r, p)$. Then, the subtree of $I_{p'}$ rooted at $x$ coincides with the tree $I_r$ in which all children of the root with weights ${\ge}c_p$ are cut, then all weights are increased by $p - p'$, and the weight of the root is set to the weight of $x$ (see Figure~\ref{fig:tree-isomorphism}).\label{lem:tree-isomorphism}
\end{lemma}
\begin{proof}
Denote by $I$ the subtree rooted at $x$. Observe that all nodes of $I$ are exactly all nodes of $I_{p'}$ associated with positions from a range $[p..p{+}i]$, for some $i < c_p$ (see Figure~\ref{fig:tree-isomorphism}). The position $p + i + 1$ is either irreducible or $c_{p+i+1} + i + 1 \ge c_p$. In order to prove the lemma, it suffices to show that (1)~all positions from $[r{+}1..r{+}i]$ are not irreducible, (2)~$c_{p+j} = c_{r+j}$, for any $j \in [1..i]$, and (3)~the position $r + i + 1$ is either irreducible or $c_{r+i+1} + i + 1 \ge c_p$.

(1) Suppose, to the contrary, that a position $r + j$, for some $j \in [1..i]$, is irreducible. Then, since $\lcp(r, p) = c_p$ and $j \le i < c_p$, we have $\lcp(r + j, p + j) = c_p - j$. Therefore, $c_{p+j}$, the length of the common prefix of $s[p{+}j..n{-}1]$ and its closest irreducible lexicographical neighbour, must be at least $c_p - j$ (since it was assumed that $r + j$ is irreducible). But then the weight $w_{p+j}$ of the node associated with $p + j$ is at least $(c_p - j) + (p + j - p') = c_p + p - p'$, which is equal to the weight $w_p = c_p + p - p'$ of $x$. Thus, $x$ cannot be an ancestor of the node, which is a contradiction.

(2) For $j \in [1..i]$, $s[p{+}j..n{-}1]$ and $s[r{+}j..n{-}1]$ share a common prefix of length $c_p - j$ since $\lcp(r, p) = c_p$ and $j \le i < c_p$. Further, $c_{p+j} < c_p - j$ since the weight $w_{p+j} = c_{p+j} + p + j - p'$ is smaller than the weight $w_p = c_p + p - p'$ of its ancestor $x$. Suppose, without loss of generality, that $s[r{+}j..n{-}1]$ is lexicographically smaller than $s[p{+}j..n{-}1]$ (the case when it is greater is analogous). Then, neither of the suffixes $s[t..n{-}1]$ lexicographically lying between $s[r{+}j..n{-}1]$ and $s[p{+}j..n{-}1]$ can start with an irreducible position, for otherwise $c_{p+j} \ge \lcp(t, p + j) \ge \lcp(r + j, p + j) \ge c_p - j$. Therefore, the closest irreducible lexicographical neighbours of $s[p{+}j..n{-}1]$ and $s[r{+}j..n{-}1]$ coincide and $c_{p + j} = c_{r + j}$ $(<c_p - j \le \lcp(p + j, r + j))$.

(3) Suppose that $r + i + 1$ is not irreducible and, by contradiction, $c_{r + i + 1} < c_p - i - 1$. The suffixes $s[p{+}i{+}1..n{-}1]$ and $s[r{+}i{+}1..n{-}1]$ have a common prefix of length $c_p - i - 1$ (note that $i + 1 \le c_p$). The position $p + i + 1$ is not irreducible since otherwise $c_{r + i + 1} \ge \lcp(p + i + 1, r + i + 1) \ge c_p - i - 1$. Then, we have $c_{p+i+1} \ge c_p - i - 1$ because otherwise the node corresponding to $p + i + 1$ would have weight $w_{p + i + 1} = c_{p + i + 1} + p + i + 1 - p'$ that is smaller than the weight $w_p = c_p + p - p'$ of $x$ and, thus, would be a descendant of $x$. Let $s[t..n{-}1]$ be the closest irreducible lexicographical neighbour of $s[p{+}i{+}1..n{-}1]$ so that $\lcp(t, p + i + 1) = c_{p + i + 1}$. Since $c_{p + i + 1} \ge c_p - i - 1$ and $\lcp(p + i + 1, r + i + 1) \ge c_p - i - 1$, we obtain $\lcp(t, r + i + 1) \ge c_p - i - 1$. Because $t$ is irreducible, we deduce that $c_{r + i + 1} \ge c_p - i - 1$, which is a contradiction.
\end{proof}

\begin{figure}[htb]
\centering
\includegraphics[scale=1]{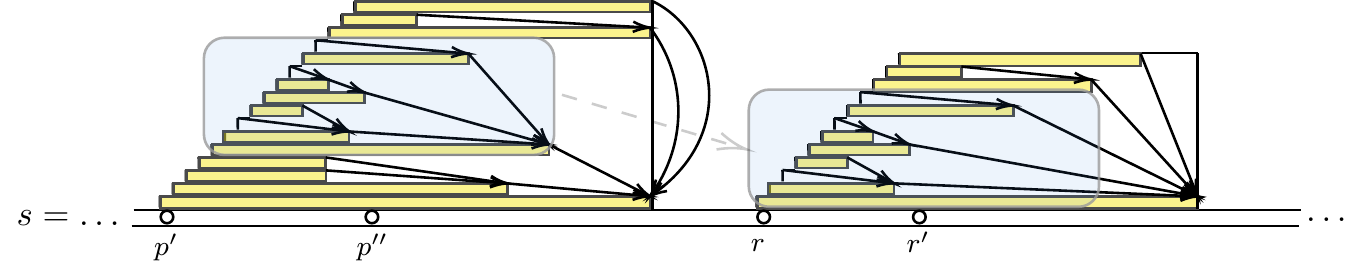}
\caption{For irreducible positions $p'$ and $r$, the subtree of $I_{p'}$ highlighted by the left blue rectangle is isomorphic to $I_r$ after cutting the children of the root of $I_r$ outside of the right blue rectangle.}\label{fig:tree-isomorphism}
\end{figure}

An array $\hat{a}$ corresponding to a heavy path in a tree $I_{p'}$, for an irreducible position~$p'$, occupies $\Oh(1 + \ell_{p'} / \log^2 n)$ space. Recall that $\ell_{p'} = \PLCP[p']$ is an irreducible LCP value since $p'$ is an irreducible position. Therefore, we can afford, for each tree $I_{p'}$, the construction and storage of the array $\hat{a}$ corresponding to the unique heavy path containing the root of~$I_{p'}$. The overall space required for this is $\Oh(n + \sum_{p'} \ell_{p'} / \log^2 n)$, where the sum is through all irreducible positions $p'$, which is upperbounded by $\Oh(n + n / \log n) = \Oh(n)$ due to Lemma~\ref{lem:irreducible-lcp}. Other heavy paths are processed as follows.

Suppose that we are to preprocess a heavy path $x_1 \to x_2 \to \cdots \to x_k$ with node weights $w_{i_1}, w_{i_2}, \ldots, w_{i_k}$ in a tree $I_{p'}$ such that $x_k$ (the node closest to the root) is not the root of~$I_{p'}$. We compute sets $S_d$ and pointers $P_h$ corresponding to the path in $\Oh(k)$ time. As for the array $\hat{a}$, we either construct it for the path from scratch or reuse a suitable array from another path; the details follow. Let $w_{i_m} = \max\{w_{i_j} \colon w_{i_j} \le \ell_{p'}\}$. As was discussed, the array $\hat{a}$ encodes, in a sense, a predecessor data structure for the weights $w_{i_1}, w_{i_2}, \ldots, w_{i_m}$. In order to have some flexibility for the reuse of arrays, we modify this scheme slightly so that sometimes $\hat{a}$ does not encode the last two weights $w_{i_{m-1}}$ and $w_{i_m}$. The algorithm that answers a predecessor query for a threshold $w \le \min\{\ell_{p'}, w_{i_k}\}$ on the path is altered accordingly: we first compare $w$ to $w_{i_{m-1}}$ and $w_{i_m}$, and only then, if necessary, use the array $\hat{a}$ as described above.

Let $x_k$ correspond to a non-irreducible position $p$ and let $s[r..n{-}1]$ be the closest irreducible lexicographical neighbour of $s[p..n{-}1]$. As was shown in Section~\ref{sec:reduction-to-irreducible}, the position $r$ can be found in $\Oh(1)$ time. By Lemma~\ref{lem:tree-isomorphism}, the subtree $I$ rooted at $x_k$ is isomorphic to the tree $I_r$ in which all children of the root with weights greater than or equal to $c_p$ are cut, then all weights are increased by $p - p'$, and the root weight is set to $w_{i_k}$. Denote this isomorphism by $\phi$. Note that $\phi(x_k)$ is the root of $I_r$. The main corollary of Lemma~\ref{lem:tree-isomorphism} is that in the subtree $I$ any edge $u\to v$ that is not incident to $x_k$ is heavy iff the corresponding edge $\phi(u) \to \phi(v)$ in $I_r$ is heavy. Therefore, all edges in the path $\phi(x_1) \to \phi(x_2) \to \cdots \to \phi(x_k)$ are heavy, except, possibly, the last edge $\phi(x_{k-1}) \to \phi(x_k)$, and no heavy edge enters the node $\phi(x_1)$ in $I_r$. By Lemma~\ref{lem:tree-isomorphism}, the weights of the nodes $\phi(x_1)$, $\phi(x_2), \ldots$, $\phi(x_{k-1})$ are $w_{i_1} - (p - p')$, $w_{i_2} - (p - p'), \ldots$, $w_{i_{k-1}} - (p - p')$, respectively.  We proceed further in three separate cases.

\subparagraph{\boldmath Heavy $\phi(x_{k-1}) \to \phi(x_k)$ and $\ell_r$ is large enough.}
Suppose that the edge $\phi(x_{k-1}) \to \phi(x_k)$ is heavy. Then, $\phi(x_1) \to \phi(x_2) \to \cdots \to \phi(x_k)$ is the unique heavy path of $I_r$ that contains the root.  Let $\hat{a}_\phi[0..c_\phi]$ be an array for predecessor queries that was explicitly stored in $\Oh(1 + \ell_r / \log^2 n)$ space for the path $\phi(x_1) \to \phi(x_2) \to \cdots \to \phi(x_k)$. By definition, for any $d \in [0..c_\phi]$, we have $\hat{a}_\phi[d] = 1$ iff, for some $j \in [1..k]$, the number $(w_{i_j} - (p - p')) - (w_{i_1} - (p - p')) = w_{i_j} - w_{i_1}$ lies in the range $[d\lfloor\log n\rfloor .. (d + 1)\lfloor\log n\rfloor{-}1]$. Since the latter is also a criterium for $\hat{a}[d] = 1$, the array $\hat{a}_\phi$ can therefore be reused to imitate $\hat{a}$ if $\hat{a}_\phi$ is sufficiently long to ``encode'' the weights $w_{i_1}, w_{i_2}, \ldots, w_{i_{m-2}}$. More precisely, this is the case iff $\ell_r \ge w_{i_{m-2}} - (p - p')$. Thus, if the edge $\phi(x_{k-1}) \to \phi(x_k)$ is heavy and $\ell_r \ge w_{i_{m-2}} - (p - p')$, then the overall preprocessing of the path $x_1 \to x_2 \to \cdots \to x_k$ takes only $\Oh(k)$ time since we can store a pointer to the array $\hat{a}_\phi$ and reuse $\hat{a}_\phi$ to imitate the array $\hat{a}$.

Unfortunately, neither of these two conditions necessarily holds in general: the edge $\phi(x_{k-1}) \to \phi(x_k)$ might be light in $I_r$ and $\ell_r$ might be less than $w_{i_{m-2}} - (p - p')$.

\subparagraph{\boldmath Light $\phi(x_{k-1}) \to \phi(x_k)$ and $\ell_r$ is large enough.}
Suppose that the edge $\phi(x_{k-1}) \to \phi(x_k)$ is light in $I_r$ and $\ell_r \ge w_{i_{m-2}} - (p - p')$. By Lemma~\ref{lem:tree-isomorphism}, the edge necessarily becomes heavy if we cut all children of the root of $I_r$ with weights greater than or equal to $c_p$. We assign numbers $1,2,\ldots$ to the children of each node in the trees $I_r$ and $I_{p'}$ according to the order in which they were attached during the construction of the trees, i.e., in the increasing order of their associated positions. It follows from the proof of Lemma~\ref{lem:tree-isomorphism} that if $x_{k-1}$ is the $h$th child of $x_{k}$, then $\phi(x_{k-1})$ is the $h$th child of the root of $I_r$. Therefore, the node $\phi(x_{k-1})$ can be located in $\Oh(1)$ time.

We associate with each child of the root of $I_r$ a pointer, initially set to null. If the pointer in $\phi(x_{k-1})$ is still null at the time we access it while preprocessing the heavy path $x_1 \to x_2 \to \cdots \to x_k$, then we create from scratch a new array $\hat{a}_\phi$ for the heavy path $\phi(x_1) \to \phi(x_2) \to \cdots \to \phi(x_{k-1})$ in $\Oh(k + \ell_r / \log^2 n)$ time and set the pointer of $\phi(x_{k-1})$ to refer to this array. Since $\ell_r \ge w_{i_{m-2}} - (p - p')$, the array $\hat{a}_\phi$ can be reused to imitate $\hat{a}$ for the path $x_1 \to x_2 \to \cdots \to x_k$ in the same way as was described above. If the pointer in the node $\phi(x_{k-1})$ is not null, then it already refers to a suitable array $\hat{a}_\phi$ that can be reused (note that $\phi(x_1) \to \phi(x_2) \to \cdots \to \phi(x_{k-1})$ is the unique heavy path of the tree $I_r$ that contains the node $\phi(x_{k-1})$). Thus, the preprocessing takes $\Oh(k)$ time plus, if necessary, $\Oh(k + \ell_r / \log^2 n)$ time required to construct a new array $\hat{a}_\phi$.

The following lemma shows that a non-null pointer to an array $\hat{a}_\phi$ might be assigned to at most $\log n$ distinct children of the root of $I_r$ (namely, those children that become connected to the root by a heavy edge after a number of children with greater weights were removed). This implies that the overall construction time for all the arrays $\hat{a}_\phi$ throughout the whole algorithm is $\Oh(n + \sum_{r}(\ell_r / \log^2 n)\log n) = \Oh(n)$, where the sum is through all irreducible positions $r$ (so that $\sum_r \ell_r = \Oh(n\log n)$ by Lemma~\ref{lem:irreducible-lcp}).

\begin{lemma}
Let $S$ be a tree with at most $n$ nodes whose root $r$ has $m$ children ordered arbitrarily. Suppose that we remove the children of the root (with the subtrees rooted at them) from right to left, one by one, thus producing trees $S_1, S_2, \ldots, S_m$. Let $z_i$ be a node of the tree $S_i$ such that $z_i$ is connected to $r$ by a heavy edge, or $z_i$ is $r$ itself if $r$ has no incident heavy edges in $S_i$. Then, the set $\{z_1, z_2, \ldots, z_m\}$ contains at most $\log n$ distinct nodes.\label{lem:heavy-after-cut}
\end{lemma}
\begin{proof}
If $z_i \to r$ is a heavy edge in $S_i$, then $\treesize(z_i) > \treesize(r) / 2$ in $S_i$. Therefore, if we cut any other child of $r$ in $S_i$, only the number $\treesize(r)$ decreases and, thus, the edge remains heavy. But when we remove $z_i$ and its subtree from $S_i$, the number $\treesize(r)$ decreases by more than half. Such halving may happen at most $\log n - 1$ times, and the result follows.
\end{proof}

\subparagraph{\boldmath Small $\ell_r$.}
Suppose that $\ell_r < w_{i_{m-2}} - (p - p')$. Let $\bar{p}$ be the position associated with the node $x_{k-1}$ in the tree $I_{p'}$ and let $s[\bar{r}..n{-}1]$ be the closest irreducible lexicographical neighbour of $s[\bar{p}..n{-}1]$. By analogy to the isomorphism $\phi$, we define using Lemma~\ref{lem:tree-isomorphism} an isomorphism $\psi$ that maps the subtree of $I_{p'}$ rooted at the node $x_{k-1}$ onto a ``pruned'' tree $I_{\bar{r}}$ in which all children of the root with weights greater than or equal to $c_{\bar{p}}$ are cut. Note that $\psi(x_{k-1})$ is the root of $I_{\bar{r}}$. By Lemma~\ref{lem:tree-isomorphism}, the weights of the nodes $\psi(x_1)$, $\psi(x_2), \ldots$, $\psi(x_{k-2})$ are $w_{i_1} - (\bar{p} - p')$, $w_{i_2} - (\bar{p} - p'), \ldots$, $w_{i_{k-2}} - (\bar{p} - p')$, respectively.

It turns out that, instead of imitating the array $\hat{a}$ for the path $x_1 \to x_2 \to \cdots \to x_k$ using an array $\hat{a}_\phi$ from the tree $I_r$, one can in quite the same way imitate $\hat{a}$ using an appropriate array $\hat{a}_\psi$ from $I_{\bar{r}}$, which must be sufficiently long in the case $\ell_r < w_{i_{m-2}} - (p - p')$. More precisely, we are to prove that $\ell_{\bar{r}} \ge w_{i_{m-2}} - (\bar{p} - p')$, which guarantees that essentially the same approach works well: if the edge $\psi(x_{k-2}) \to \psi(x_{k-1})$ is heavy, then $\hat{a}_\psi$ is an array associated with the unique heavy path containing the root of $I_{\bar{r}}$ and, since $\ell_{\bar{r}} \ge w_{i_{m-2}} - (\bar{p} - p')$, the array $\hat{a}_\psi$ is long enough to imitate $\hat{a}$; if the edge $\phi(x_{k-2}) \to \psi(x_{k-1})$ is light, then we either reuse a suitable array $\hat{a}_\psi$ that was already stored in the child $\psi(x_{k-2})$ of the root of $I_{\bar{r}}$ (again $\hat{a}_\psi$ can imitate $\hat{a}$ since $\ell_{\bar{r}} \ge w_{i_{m-2}} - (\bar{p} - p')$) or we create this array $\hat{a}_\psi$ from scratch for the path $\psi(x_1) \to \psi(x_2) \to \cdots \to \psi(x_{k-2})$ in $\Oh(k + \ell_{\bar{r}} / \log^2 n)$ time and store a pointer to it in $\psi(x_{k-2})$. We do not go further into details since they are essentially the same as above.

We actually prove a stronger condition $\ell_{\bar{r}} \ge w_{i_{m-1}} - (\bar{p} - p')$, which implies $\ell_{\bar{r}} \ge w_{i_{m-2}} - (\bar{p} - p')$ since $w_{i_{m-1}} > w_{i_{m-2}}$. Recall that $w_{i_{m-1}} = c_{\bar{p}} + \bar{p} - p'$. Hence, the stronger condition is equivalent to $\ell_{\bar{r}} \ge c_{\bar{p}}$. If the suffix $s[\bar{r}..n{-}1]$ is lexicographically larger than $s[\bar{p}..n{-}1]$, then we immediately obtain $\ell_{\bar{r}} = \plcp[\bar{r}] \ge \lcp(\bar{p}, \bar{r}) = c_{\bar{p}}$. Assume that $s[\bar{r}..n{-}1]$ is lexicographically smaller than $s[\bar{p}..n{-}1]$.
Let us show that this case is actually impossible since it contradicts the condition of ``small $\ell_r$'' that was assumed in the beginning: $\ell_r < w_{i_{m-2}} - (p - p')$.
Denote $j = \bar{p} - p$. Since $\lcp(p, r) = c_p$ and $j < c_p$, we obtain $\lcp(\bar{p}, r + j) = c_p - j$. Further, $c_{\bar{p}} < c_p - j$ because the weight $w_{i_{m-1}} = c_{\bar{p}} + \bar{p} - p'$ of the node $x_{m-1}$ is less than the weight $w_{i_{k}} = c_p + p - p'$ of $x_k$. Since $\lcp(\bar{p}, \bar{r}) = c_{\bar{p}} < c_p - j = \lcp(\bar{p}, r + j)$, we obtain $\lcp(\bar{r}, r + j) = c_{\bar{p}} < \lcp(\bar{p}, r + j)$ and, further, since we assumed that $s[\bar{r} ..n{-}1]$ is lexicographically smaller than $s[\bar{p}..n{-}1]$, the suffix $s[\bar{r}..n{-}1]$ is lexicographically smaller than $s[r{+}j..n{-}1]$. Hence, $\ell_{r + j} \ge \lcp(\bar{r},r + j) = c_{\bar{p}}$. It follows from Lemma~\ref{lem:inherent-arrays} that $\ell_r = \ell_{r + j} + j$. Then, $\ell_r \ge c_{\bar{p}} + j$. Conversely, we deduce $\ell_r < w_{i_{m-2}} - (p - p') < w_{i_{m-1}} - (p - p')$ and, substituting $w_{i_{m-1}} = c_{\bar{p}} + p + j - p'$, we obtain $\ell_r < c_{\bar{p}} + j$, which is a contradiction.

\medskip

To sum up, we spend linear time preprocessing each heavy path in every tree $I_{p'}$ plus $\Oh(n)$ total time to construct arrays $\hat{a}_\phi$ and $\hat{a}_\psi$, each of which is associated with one of $\log n$ particular children of the root in one of the trees (according to Lemma~\ref{lem:heavy-after-cut}). Therefore, since all the trees $I_{p'}$ contain $n$ nodes in total, the overall time is $\Oh(n)$.

\section{Concluding Remarks}
\label{sec:conclusion}

We believe that the presented solution, while certainly still quite complicated and impractical, is more implementable than that of~\cite{GLN}. We expect that, with additional combinatorial insights, one can simplify it even further, perhaps eventually arriving at a practical data structure for the problem. A number of natural and less vaguely formulated problems also arise. For example, can we maintain the weighted ancestor data structure during an online construction of the suffix tree? Is it possible to reduce the space usage and support weighted ancestors in compact suffix trees~\cite{MunroNavarroNekrich,Sadakane} with $\Oh(n)$ or $\Oh(n\log\sigma)$ additional \emph{bits} of space, where $\sigma$ is the alphabet size?

To the best of our knowledge, the lemma about the sum of irreducible LCP values has found relatively few applications other than in the construction of LCP and PLCP arrays (e.g., see references in~\cite{KarkkainenKempaPiatkowski}):~\cite{BadkobehEtAl2,BadkobehEtAl,KempaKociumaka2,ValimakiPuglisi} (the application in~\cite{KempaKociumaka2}, however, is somewhat spectacular). The techniques developed in our paper might be interesting by themselves and pave the way to more applications of irreducible LCPs in algorithms on strings.



\bibliography{suff-tree}

\end{document}